\documentclass[a4paper,draft]{amsproc}
\usepackage{amssymb}
\usepackage{amscd} 
\usepackage{mathrsfs}

 \theoremstyle{plain}
 \newtheorem{thm}{Theorem}[section]
 \newtheorem{prop}{Proposition}[section]

 \newtheorem{cor}{Corollary}[section]
 \theoremstyle{definition}
 \newtheorem{exm}{Example}[section]
 \newtheorem{dfn}{Definition}[section]
 \theoremstyle{remark}
 
 \numberwithin{equation}{section}

 \def\A{{\bf A}}
 \def\B{{\bf B}}
 
 \def\O{{\bf \Omega}}
 \def\vv{{\hat v}}
 \def\phi{{\hat \varphi}}
 \def\phiA{{\hat \varphi}^{\bf A}}
 \def\Sv{{v_1,\ldots,v_n}}
 \def\Svm{{v_1,\ldots,v_m}}
 \def\Svv{{{\hat v}_1,\ldots,{\hat v}_n}}
 
 \def\Bb{{\bf 2}}
 \def\BbI{{{\bf 2}^I}}
 \def\AI{{{\bf A}^I}}
 \def\K{\mathcal{K}}
 \def\J{\mathcal{J}}
 \def\P{\mathcal{P}}
 
 \def\X{\mathcal{X}}
 \def\Y{\mathcal{Y}}
 \def\PP{\mathscr{P}}
 \def\L{\mathscr{L}}
 \def\M{\mathfrak{M}}
 
 \def\Aut{{\rm Aut}}

\renewcommand{\leq}{\leqslant}
\renewcommand{\geq}{\geqslant}

\title[Computing finite models using free Boolean generators]{Computing finite models using free Boolean generators}



\author[Mijajlovi\'c]{\bfseries \v Zarko Mijajlovi\'c }
\author[Pejovi\'c]{\bfseries   Aleksandar Pejovi\'c}

\address{
Faculty of Mathematics \\ 
University of Belgrade   \\ 
Belgrade\\
Serbia}
\email{zarkom@matf.bg.ac.rs}
\address{
Institute of Mathematics \\ 
Serbian Academy of Science and Arts   \\ 
Belgrade\\
Serbia}
\email{pejovica@mi.sanu.ac.rs}

\thanks{Partially supported by Serbian Ministry of Science, Grant III 44006} 

\begin{document}

\vspace{18mm}
\setcounter{page}{1}
\thispagestyle{empty}

\begin{abstract}
  A parallel method for computing Boolean expressions based on the properties of
  finite free Boolean algebras is presented.
  We also show how various finite combinatorial objects can be codded in the formalism of
  Boolean algebras and counted by this procedure.
  Particularly, using a translation of first order predicate formulas
  to propositional formulas, we give a method for constructing and
  counting   finite models of the first order theories.
  An  implementation of the method that can be run on
  multi-core CPUs as well as on highly parallel GPUs is outlined.
\end{abstract}

\maketitle

\section{Introduction}

  Even ordinary personal computers are capable
  for specific  massive parallel computations.
  Examples of this kind are logical operations which can be
  computed bitwise, i.e., by use of all register bits in one processor cycle.
  Based on this idea, we propose a method for computing Boolean expressions using
  the parallel structure of standard computer processors.
  The mathematical background of our approach is based on the properties of finite free Boolean algebras.
  The idea of parallelization of computing logical operations in this way is indicated in \cite{1}.
  The basic idea is as follows.

  Let $f(x_1,x_2,\ldots,x_n)$ be a Boolean expression in $n$ variables
  $x_1,x_2,\ldots,x_n$.  We give a construction of $n$ Boolean vectors $b_1,b_2,\ldots,b_n$
  of  size $2^n$ with the following property:
  \medskip

  ($\PP$)\quad  $f(b_1,b_2,\ldots,b_n )$ is a Boolean vector that codes the full DNF of  $f$.
  \smallskip

  It appears that  vectors  $b_1,b_2,\ldots,b_n$ are exactly  free
  generators of a free Boolean algebra having $n$ free  generators.

  Using a translation procedure from the first order predicate formulas
  to propositional formulas, we give a method for constructing and
  counting various combinatorial objects. This idea is formally developed
  in \cite{2}, but it was used there in the study of problems in the infinitary
  combinatorics, particularly in finding their complexity in the
  Borel hierarchy.
  Related combinatorial problems are considered, for example
  the number of automorphisms of finite structures and various partition problems over finite sets.
  We also give an  implementation of the method that can be run on
  multi-core CPUs as well as on highly parallel GPUs (Graphics processing units). 

  Standard notation and terminology from model theory is assumed
  as in [\ref{Keisler}] and  [\ref{MijajloMT}].
  Also, for notions from universal algebras we shall refer to [\ref{Burris}].
  Models of a first order language
  $L$ are denoted by bold capital letters
  $\bf A$, $\bf B$, etc, while their domains respectively by $A$, $B$ and so on.
  By a domain we mean any nonempty set.
  The letter $L$ will be  used to denote a first-order language.
  The first order logic is denoted by $L_{\omega\omega}$ and the
  propositional calculus with a set $\mathcal{P}$ of propositional variables by $L_{\omega}^\mathcal P$,
  or simply $L_{\omega}$.
  The set of natural numbers $\{0,1,2,\ldots\}$ is denoted by $N$.
  We also take $2=\{0,1\}$. By $\Bb$ we denote the two-element Boolean algebra
  and then $\BbI$ is the power of $\Bb$, while $\bf 0$ and $\bf 1$ are
  respectively  the smallest and the greatest element of $\BbI$.
  Occasionally elements of  $\BbI$ are called Boolean vectors.
  Whenever is  needed to distinguish the formal equality sign from identity,
  for the first one we shall keep $=$, while $\equiv$ denotes identity.

  \section{Variables}

  In this section we develop and explain the logical and algebraic
  background for our computing method. The power of a model $\A$, the product
  $\prod_{i\in I}\A$, is denoted by $\A^I$.

  \subsection{Interpretation of variables}

  By a set of variables we mean any non\-empty set $V$ so that
  no $v\in V$ is a finite sequence of other elements from $V$. This assumption
  secures the unique readability of terms and formulas.
  Particularly we shall consider finite and countable sets of variables $V$,
  e.g. $V= \{v_0,v_1,\ldots \}$.
  A valuation of a domain $A$ is any map from $V$ to $A$.
  Let $I$ denote the set of all valuations from domain $A$,
  i.e., $I= A^V$. In this section, the letter $I$ will be reserved for the set
  of valuation of a domain $A$.  Sometimes we shall assume that elements from $I$ will have finite
  supports.

  \begin{dfn}\label{IV} ({\it Interpretation of variables}).
  Let $v$ be a variable from $V$. The interpretation of  variable $v$ in domain $A$ is the map
  $\hat v \colon I\to A$ defined by  $\hat v(\mu)= \mu(v)$, $\mu\in I$.
  \end{dfn}

  The set of interpretations of variables from $V$ into domain $A$ is denoted by $\hat V_A$.
  Therefore, $\hat V_A= \{\hat v\colon v\in V \}$.

  Let $\varphi(v_1,\ldots,v_n)$ be a formula of a language $L$ having free variables
  $v_1,\ldots,v_n$ and $\bf A$ a model of $L$. The map $\hat\varphi^{\bf A}(\hat v_1,\ldots,\hat v_n)$,
  abbreviated  by  $\hat\varphi^{\bf A}$, is $\hat\varphi^{\bf A}\colon I\to 2$ defined by
  $\hat\varphi^{\bf A}(\mu) = 1$ if $\bf A \models \varphi[\mu]$, otherwise
  $\hat\varphi^{\bf A}(\mu) = 0$, $\mu\in I$.  Hence $\phiA\in 2^I$.

  \begin{prop}\label{prop1}
    Let $\varphi$ be an identity $s=t$, where $s$ and $t$ are terms of $L$.
    Then the following are equivalent:
    \begin{equation}
       1^\circ\, \A^I \models \varphi[\Svv],\quad
       2^\circ\, \phiA(\Svv)= {\bf 1}, \quad
       3^\circ\, \A \models \varphi[\mu],\, \mu\in I.
    \end{equation}
  \end{prop}

  \begin{proof}
  The equivalence of $2^\circ$ and $3^\circ$ follows immediately by definition \ref{IV}.
  From $3^\circ$ follows $1^\circ$ since identities are preserved under products of
  models. Finally, assume $1^\circ$. Then
  \begin{equation}\label{eq1}
    s^\AI(\Svv) = t^\AI(\Svv).
  \end{equation}
  Let $\pi_\mu\colon \AI \to \A$ be a projection, $\mu \in I$. Since $\pi_\mu$ is
  a homomorphism we have
  \begin{equation}
    \begin{array}{rll}
       \pi_\mu(s^\AI(\Svv)) &= &s^\A(\pi_{\mu}\vv_1,\ldots,\pi_{\mu}\vv_n)   \\
                            &= &s^\A(\vv_1(\mu),\ldots,\vv_n(\mu))           \\
                            &= &s^\A(\mu(v_1),\ldots,\mu(v_n)) = s^\A[\mu].
    \end{array}
  \end{equation}
  Hence, $3^\circ$ follows by \ref{eq1}.
   \end{proof}

  For an algebra $\A$ of $L$ let ${\J}(\A)$ be the set of all identities that are
  true in $\A$. Similarly, ${\J}(\K)$ denotes the set of all
  identities that are true in all algebras of a class $\K$ of algebras of $L$.
  If ${\J}(\A) = {\J}(\B)$, $\A$ and $\B$ are algebras of $L$, we hall also write
  $\A \equiv_\J \B$.

  The notion of interpretation of variables will play the fundamental role in our
  analysis and program implementation. But they can be useful in other cases, too.
  For example, for so introduced notions it is easy to prove the Birkhoff HSP theorem
  and other related theorems. Here we prove a theorem on the
  existence of free algebras.
  The novelty of these proof is  that it does not use the notion of a term algebra
  (absolutely free algebra).
  For the simplicity of exposition, we shall assume that $L$ is countable.

  \begin{thm}\label{Birkhoff} (G. Birkhoff)
    Let $\K$ be a nontrivial abstract\footnote{closed for isomorphic images} class
    of algebras of  L, closed under
    subalgebras and products. Then $\K$ has a free  algebra over every nonempty set.
  \end{thm}

  \begin{proof}
    It is easy to see, for example by use of the downward Skolem-L\"owen\-heim theorem, that
    for each algebra $\A \in \K$ there is at most countable subalgebra $\A'$ of $\A$
    so that $\A' \equiv_\J \A$. The algebra $\A'$ is obviously isomorphic to
    an algebra of which the domain is a subset of $N$. Hence, there is a set
    $\K'= \{\A_s\colon s\in S\}$
    of  at most countable algebras such that $\K'\subseteq \K$ and $\J(\K) = \J(\K')$.

    Let $\A = \prod_s \A_s$ be the product of all algebras from $\K'$.
    Since $\K$ is closed under products, it follows $\A\in \K$, hence
    $\J(\K)\subseteq \K(\A)$. On the other hand, for each $s\in S$, $\A_s$ is
    a homomorphic image of $\A$, as $\A_s = \pi_s\A$. Hence each identity $\varphi$
    of $L$ which holds on $\A$ is also true in all algebras from $\K'$ and
    therefore in all algebras from $\K$. So we proved
    \begin{equation}\label{JKA}
      \J(\K) = \J(\A).
    \end{equation}

    Since $\K$ is nontrivial, it must be $|A|\geq 2$. Let $X$ be any non empty set.
    For our purpose we may identify $X$ with $\hat V_A$ for some set of variables $V$.
    Let $\bf \Omega$ be subalgebra of $\A^I$ generated by $\hat V_A$.
    Since $\K$ is closed under subalgebras, it follows $\O \in \K$.
    Now we prove that
    $\bf \Omega$ is a free algebra over $\hat V_A$ for class $\K$.
    Let $\B\in \K$ be an arbitrary algebra and $g \colon \hat V_A \to B$.
    Each element $a\in \Omega$ is of the form
    $a= s^{\bf \Omega}(\vv_1,\ldots,\vv_n)$ for some $L$-term $s$ and some (different)
    variables $v_1,\ldots,v_n\in V$.
    We extend $g$ to $f\colon \O \to \B$ taking
    \begin{equation}\label{fg}
       f(a)= s^{\B}(g\vv_1,\ldots,g\vv_n).
    \end{equation}
    The map $f$ is well defined.  Indeed, suppose that for some other term $t$ of $L$,
    $a= t^{\bf \Omega}(\vv_1,\ldots,\vv_n)$. Let $\varphi$ denote the
    identity $s(v_1,\ldots,v_n)= t(v_1,\ldots,v_n)$.
    Then
    $s^{\O}(\vv_1,\ldots,\vv_n) = t^{\O}(\vv_1,\ldots,\vv_n)$ and as
    $\O\subseteq \A^I$ it follows $\AI\models \varphi[\vv_1,\ldots,\vv_n]$.
    By Proposition \ref{prop1} it follows that the identity $\varphi$ holds on $\A$.
    By \ref{JKA} then $\varphi$ is true in all algebras from $\K$. Hence
    \begin{equation}
       s^{\B}(g\vv_1,\ldots,g\vv_n) = t^{\B}(g\vv_1,\ldots,g\vv_n),
    \end{equation}
    and thus we proved that the $f$ is well-defined.

    In a similar manner we prove that $f$ is a homomorphism.
    For simplicity, suppose $\ast$ is a binary operation of $L$.
    We denote the  interpretations of $\ast$ in $\O$ and $\B$ by $\cdot$.
    Take $a,b\in\Omega$ and let $s$ and $t$ be terms of $L$
    so that
    \begin{equation}
      a= s^{\O}(g\vv_1,\ldots,g\vv_n),\quad b= t^{\O}(g\vv_1,\ldots,g\vv_n)
    \end{equation}
    and let $w$ be the combined term $w= s \ast t$. Then
    \begin{equation}
      f(a\cdot b) = f(w^\O(\vv_1,\ldots,\vv_n))= w^\B(g\vv_1,\ldots,g\vv_n)= g(a)\cdot g(b).
    \end{equation}
    Thus, $f$ is a homomorphism from $\O$ to $\B$ which extends $g$.
    \end{proof}

    Suppose $\K$ is the class of algebras to which refer the previous theorem.
    We note the following.

    {\it Note} \ref{Birkhoff}.1\, It is easy now to prove the Birkhoff HSP theorem. Assume   $\K$
    is also closed under  homomorphic images and let  $T=\J(\K)$.
    Let $\O$ be a free algebra of $\K$ with infinitely many free generators.
    Then $\J(\O) = \J(\K) =T$. Suppose $\B$ is a model of $T$
    and let $\Omega$ be a free algebra for class $\K$
    and $X$ is a set of free generators of $\O$ such that $|X|\geq |B|$.
    Let $g\colon X\to B$ so that $g(X)=B$.
    Then by the same construction as in the previous proof g extends to
    some homomorphism $f\colon \O \to \B$, thus $\B$ is a homomorphic image of $\O$.
    Hence $\B$ belongs to $\K$.

    {\it Note} \ref{Birkhoff}.2\,  Assume $\A\in \K$ is an arbitrary algebra which satisfies
    condition \ref{JKA}.   Such an algebra $\A$ will be called the
    characteristic algebra for the class $\K$. By close inspection of
    the proof of Theorem \ref{Birkhoff}, we see that this condition suffices
    to construct a free algebra for $\K$ from $\A$ as we did in the
    proof of \ref{Birkhoff}.
    This idea is indicated to some extent in [\ref{Burris}], (Part II, chapter 11,
    particularly see problem 11.5, p. 77) but under  stronger   and
    amended assumptions and without referring to variable interpretations.

  \subsection{Free Boolean vectors}

  It is well known that finite free Boolean algebras with $n$ free generators
  are the algebras $\Bb^{2^n}$.  We remark that this immediately follows
  by note \ref{Birkhoff}.2, since $\Bb$ is
  the characteristic algebra for the class of all Boolean algebras.
  The structure and properties of free
  Boolean vectors of $\O_n= \Bb^{2^n}$ are discussed in [\ref{MijajloFB}] in details.

  We remind that a collection $\{b_1,\ldots,b_n\}$ of elements of a Boolean algebra $\B$ is independent
  if $b_1^{\alpha_1}\wedge \ldots \wedge b_n^{\alpha_n}\not=0$, where
  $b^1= b$ and $b^0= b'$. A similar definition of independence is for families of subsets of a given set.
  A collection $\{b_1,\ldots,b_n\}$ generates the free subalgebra of   $\B$
  if and only if it is independent, cf. [\ref{Sikorski}].
  The number of free generating sets of $\O_n$ is found in [\ref{MijajloFB}].
  In fact,  the following holds.
  \begin{thm}
    Let $S= \{1,2,\ldots,2^n\}$ and $a_n$, $b_n$, $c_n$ be the sequences defined as follows.
    \begin{enumerate}
      \item $a_n$ = number of labeled Boolean algebras with domain $S$
                    {\rm (}number of different Boolean algebras with domain $S$ {\rm )}.
      \item $b_n$ = number of independent collections $\{P_1,\ldots,P_n\}$
                    of subsets of $S$.
      \item $c_n$ = number of free generating sets $\{b_1,\ldots,b_n\}$ of $\O_n$.
    \end{enumerate}
    Then $a_n= b_n = c_n = (2^n)!/n!$.
  \end{thm}

  \begin{proof}
    The number of  labelings of
    a finite model $\A$ of size $m$ is equal to $m!/|\Aut(\A)|$.
    As $\Aut(\Bb^n)$ is isomorphic to the permutation group $S_n$,
    it follows $a_n= (2^n)!/n!$.

    Let $\B= \Bb^n$ and $\B^l$ a labeled algebra obtained from $\B$.
    Algebra $\B$ has exactly $n$ ultrafilters and so has $\B^l$.
    Let $U(\B)$ be the set of all ultrafilters of $\B$.
    By Theorem 2.2.7 in [\ref{MijajloFB}], $U(\B)$ is an independent collection of
    subsets of $S$. The map $U$ which assigns  $U(\B^l)$ to $\B^l$ is $1-1$.
    Indeed, let us for
    $S_1,\ldots,S_n\subseteq S$ and $\alpha\in 2^n$ define
    \begin{equation}
       S^\alpha = S^{\alpha_1}\cap\ldots\cap S^{\alpha_n}.
    \end{equation}

    For $a\in S$ let $P_1,\ldots,P_k, P_{k+1},\ldots,P_n\in U(\B^l)$ be such that
    $a\in P_1,\ldots,P_k$ and $a\not\in P_{k+1},\ldots,P_n$. Then
      $P_1\cap\ldots,\cap P_k\cap P_{k+1}^c\ldots\cap P_n^c =\{a\}$.

    Therefore, we proved that for each $a\in B^l$ there is a unique $\alpha\in 2^n$
    such that $P^\alpha = \{a\}$, $P_1,\ldots,P_n\in U(\B^l)$.
    Let $\wedge^l$  and $'^l$ be Boolean operations of $\B^l$.
    Then for $a, b\in B^l$ and corresponding $\alpha,\beta\in 2^n$
    we have
    \begin{equation}\label{lB}
      P^{\alpha'} = \{a^{'^l}\},\quad  P^{\alpha \wedge \beta}= \{a\wedge^l b\}
    \end{equation}
    where $\alpha'$, $\alpha \wedge \beta$ are computed in $\Bb^n$.
    Thus, we proved that $U(\B^l)$ uniquely determines $\B^l$, hence $a_n \leq b_n$.

    Suppose $P= \{P_1,\ldots,P_n\}$ is an independent
    collection of subsets of $S$. Then $P$ can serve as  $U(\B^l)$
    for certain labeled Boolean algebra $\B^l$. To prove it,
    note that each $P^\alpha$ has at least one element and that
    $\bigcup_{\alpha\in 2^n} P^\alpha$ has at most $2^n$ elements.
    This shows that $P^\alpha$ is one-element set. Therefore,
    a Boolean algebra $\B^l$ with domain $S$ is defined by \ref{lB} and
    it is easy to see that $P= U(\B^l)$. Hence $a_n=b_n$.

    Finally, as noted, a collection $X= \{X_1,\ldots,X_n\}$ of subsets of $S$
    freely generates the power set algebra $P(S)$ if and only if
    $X$ is independent. Hence, $c_n=b_n$.
    \end{proof}

    We will be dealing particularly with free generators of $\O_n$ of the following form.
    Let $a_i$, $i = 0, 1,\ldots, 2^n - 1$, be  binary expansions of  integers $i$
    with zeros  padded to the left up to the length $n$. Let $M$ be the matrix whose columns
    are vectors $a_i$. As noted in [\ref{MijajloFB}], binary vectors
    $b_i$, $i = 1, 2\ldots n$, formed by rows of $M$ are free vectors of
    $\O_n$. In the case n = 3, the matrix M and vectors $b_i$ are
    \begin{equation}
      M= \left[
         \begin{array}{cccccccc}
            0 &0 &0 &0 &1 &1 &1 &1  \\
            0 &0 &1 &1 &0 &0 &1 &1  \\
            0 &1 &0 &1 &0 &1 &0 &1
         \end{array} \right],
    \end{equation}
    $b_1 = 00001111$, $b_2 = 00110011$, $b_3 = 01010101$.

    \subsection{Computing Boolean expressions}\label{CBe}

    Let $t= t(\Sv)$ be a Boolean expression in variables $\Sv$ and $b_1,\ldots,b_n$ free generators of
    $\O_n$.
    \begin{prop}\label{prop22}
      $t^{\O_n}(b_1,\ldots,b_n)$ codes the the full DNF of $t$.
    \end{prop}

    \begin{proof}
      By our previous discussion, we may take $b_i=\vv_i$ and $I= \{\Svv\}$.
      Let $\pi_\mu$ be a projection from $\O_n$ to $\Bb$, $\mu\in I$, and $d= t^{\O_n}(b_1,\ldots,b_n)$.
      Then
      \begin{equation*}
        \pi_\mu d= \pi_\mu t^{\O_n}(\Svv) = t^\Bb(\mu(v_1),\ldots,\mu(v_n)),
      \end{equation*}
      hence
      $\displaystyle t= \sum_{\pi_\mu d= 1} v_1^{\mu_1}\cdots v_n^{\mu_n}$, so $d$ codes the full DNF of $t$.
    \end{proof}

    The parallel algorithm for computing $d= t^{\O_n}(b_1,\ldots,b_n)$
    is described in details in [\ref{MijajloFB}], Section 2.
    We repeat in short this procedure. Suppose we have a $2^k$-bit processor at our
    disposal, $k<n$. Each vector $b_i$ is divided into $2^{n-k}$  consecutive
    sequences of equal size. Hence,   $b_i$ consists of $2^{n-k}$ blocks $b_{ij}$, each of   size $2^k$.
    To find $d$, blocks   $d_j= t(b_{1j},\ldots,b_{nj})$ of size $2^k$ are computed bitwise
    for $j= 1,2,\ldots,2^{n-k}$. Then the combined vector $d_1d_2\ldots d_{2^{n-k}}$ is
    the output vector $d$. The total time for computing $d$ approximately is
    $T= 2^{l+n-k}\delta$, where $2^l$ is the total number of nodes in the
    binary expression tree of the term $t$ and $\delta$ is the time interval
    for computing bitwise one logical operation\footnote{For modern computers,
    $\delta\approx 10^{-9}$ seconds}.

    Suppose now that we have $2^r$  $2^k$-bit processors. Computations of $d_j$ is distributed
    among all processors and they compute $t(b_{1j},\ldots,b_{nj})$ in parallel.
    Actually, they are acting as a single $2^{k+r}$-bit processor. Hence, the total time
    for computing $d$ in this case is $T= 2^{l+n-k-r}\delta$.

    We implemented this  algorithm on a PC with two GPU's, each having
    $2^{11}$ 32-bit processors. Therefore, this installation is equivalent to
    a machine with  one $2^{17}$ - bit processor, as $k=5$ and $r=12$.
    Our implementation at this moment is based on 30 free Boolean vectors,
    each with $2^{30}$ bits. This implementation theoretically computes a Boolean term
    $t$ with 30 variables and $2^{17}$ nodes in it's Boolean expression in time $2^{30}\delta$
    i.e., in about one second. Our experimental results are very close to this time.

    We note that the number of free variables is limited by  the size of internal memory and the size
    of the output vector $d$. The installation that we are using could admit
    the described computation with 35 free Boolean vectors.  With further partition
    of the particular problems the computation can be done in real time with up to 50 Boolean variables.
    For the most powerful modern supercomputers, these numbers respectively are 50 and 70.
    It is interesting that  these numbers were anticipated in [\ref{MijajloFB}], 15 years ago.

    \section{Computing finite models}

    Using a translation from $L_{\omega\omega}$ to $L_{\omega}$,
    we are able to state and computationally solve various problems on finite structures. There are
    attempts of this kind.  For example H. Zhang developed the system SATO for computing specific
    quasigroups, see [\ref{Zhang}]. There are many articles  with the similar approach on
    games, puzzles and design of particular patterns.
    An example of this kind is Lewis article [\ref{Lewis}] on Sudoku.

   \subsection{Translation from $L_{\omega\omega}$ to $L_\omega$}\label{Translation}

   A method for coding some notions, mostly of the combinatorial nature and
   related to countable first-order structures,
   by theories of propositional calculus $L_{\omega_1}$ is presented in [\ref{MijajloBorel}].
   The primary goal there was  to study the complexity of these notions in Borel hierarchy.
   The coding is given there by a map $\ast$. We reproduce this map adapted for our needs.

   Let $L$ be a finite first-order language and $L_A = L \cup \{\underline{a} | a \in A\}$,
   where $A$ is a finite non-empty set.
   Here  $\underline a$ is a new constant symbol, the name of the element $a$.
   We define the set $\P$   of propositional letters as follows
   \begin{equation}\label{P}
     \begin{array}{rll}
       \P \hskip -2mm&= &\{p_{Fa_1\ldots a_k b} |\, a_1,\ldots,a_k, b \in A, F \textrm{ is a $k$-ary function symbol of } L\}\,  \cup  \\
          \hskip -2mm&  &\{q_{Ra_1\ldots a_k b} |\, a_1,\ldots,a_k, b \in A, R \textrm{ is a $k$-ary relation symbol of } L\}
     \end{array}
   \end{equation}

   The map $\ast$ from the set Sent${}_{L_A}$ of all $L_{\omega\omega}$-sentences of $L_A$
   into the set of  propositional formulas of $L_{\omega}^\P$ is defined recursively as follows.
   \begin{equation}
     \begin{array}{c}
       \phantom{A} \hskip -0.9cm
       (F(\underline{a}_1,\ldots, \underline{a}_k) = \underline{b})^\ast \equiv  p_{Fa_1\ldots a_k b},\quad
                 (R(a_1,\ldots, a_k))^\ast     \equiv  q_{Ra_1\ldots a_k},                        \\
       \phantom{A} \hskip -5.23cm
       (F(\underline{a}_1,\ldots, \underline{a}_k) = F'(\underline{a}_1',\ldots, \underline{a}_k'))^\ast\equiv      \\
            \bigwedge_{b\in A}
            (F(\underline{a}_1,\ldots, \underline{a}_k)= \underline{b})^\ast \Rightarrow
            (F'(\underline{a}_1',\ldots, \underline{a}_k')= \underline{b})^\ast),                 \\
       \phantom{A} \hskip -3cm
       (F(t_1(\underline{a}_{11},\ldots, \underline{a}_{1m}),\ldots, t_k(\underline{a}_{k1},\ldots,\underline{a}_{km}))= \underline{b})^\ast \equiv    \\
                 \bigwedge_{(b_1,\ldots,b_k)\in A^k}\left(
                 \bigwedge_{i=1}^k (t_i(\underline{a}_{i1},\ldots, \underline{a}_{im})= \underline{b}_i)^\ast \Rightarrow p_{Fb_1\ldots b_k b}\right), \\
       \phantom{A} \hskip -3.66cm
       (R(t_1(\underline{a}_{11},\ldots, \underline{a}_{1m}),\ldots, t_k(\underline{a}_{k1},\ldots, \underline{a}_{km})))^\ast \equiv                  \\
                 \bigwedge_{(b_1,\ldots,b_k)\in A^k}\left(
                 \bigwedge_{i=1}^k (t_i(\underline{a}_{i1},\ldots, \underline{a}_{im})=\underline{b}_i)^\ast \Rightarrow q_{Rb_1\ldots b_k b}\right),  \\
       (\neg \varphi)^\ast  \equiv \neg\varphi^\ast,\quad
                 (\varphi \wedge \psi)^\ast \equiv \varphi^\ast \wedge \psi^\ast,\quad
                 (\varphi \vee \psi)^\ast \equiv \varphi^\ast \vee \psi^\ast,         \\
       (\forall x \varphi(x))^\ast \equiv \bigwedge_{a\in A} \varphi(\underline{a})^\ast,\quad
                 (\exists x \varphi(x))^\ast \equiv\bigvee_{a\in A} \varphi(\underline{a})^\ast.
     \end{array}
   \end{equation}

   The constants symbols from $L$ are handled in this definition of $\ast$ as $0$-placed function symbols.
   If $L$ has only one function symbol $F$, then we shall write $p_{b_1\ldots b_k b}$ instead of $p_{Fb_1\ldots b_k b}$.
   The similar convention is assumed for a relation symbol $R$. For example, if $\varphi$ is the sentence
   which states the associativity of the binary function symbol $\cdot$, it is easy to see that
   the $\ast$-transform of $i\cdot j =u$ is   $p_{iju}$  and that over domain
   $I_n= \{0,1,\ldots,n-1\}$, $\varphi^\ast$ is equivalent to
   \begin{equation}
      \bigwedge_{i,j,k,u,v,l < n} ((p_{iju} \wedge p_{jkv}\wedge p_{ukl}) \Rightarrow p_{ivl})
   \end{equation}

   If not stated otherwise, we assume that the domain of a finite model $\A$ having $n$ elements is $I_n= \{0,1,\ldots, n-1\}$.
   Observe that $\P$ is finite. If $\A$ is a model of $L$, note that the simple expansion
   $(\A, a)_{a \in A}$ is a model of $L_A$.

   \subsection{Correspondence between models of $T$ and $T^\ast$}\label{correspondence}

   Using  translation $\ast$, we give a method for
   constructing and counting finite models of first order theories for a finite language $L$.
   In the rest of the paper the notion of a labeled model will have the
   important role. Therefore we fix this and related concepts.

   Let $\A$ be a finite model of $L$, $|A|= n$. Any one-to-one and onto map
   $\alpha\colon I_n\to A$ will be called the labeling of $\A$.
   We can transfer the structure of $\A$ to a model $\A_\alpha$ with the domain $I_n$
   in the usual way:
   \begin{itemize}
      \item[1.]
        If $R\in L$ is is a $k$-placed relation symbol then we take  \\
        \phantom{A}\hskip 5mm  $R^{\A_\alpha}(i_1,\ldots,i_k)$ iff $R^{\A}(\alpha(i_1),\ldots,\alpha(i_k))$, $i_1\ldots,i_k\in I_n$.
      \item[2.]
        If $F\in L$ is is a $k$-placed function symbol then we take  \\
        \phantom{A}\hskip 5mm  $F^{\A_\alpha}(i_1,\ldots,i_k)= \alpha^{-1}(F(\alpha(i_1),\ldots,\alpha(i_k))$, $i_1\ldots,i_k\in I_n$.
      \item[3.] If $c\in L$ is a constant symbol then  $c^{\A_\alpha}= \alpha^{-1}(c^A)$.
   \end{itemize}

   We see that $\alpha\colon \A_\alpha \cong \A$. We shall call $\A_\alpha$ a labeled model of $\A$.
   Let $c_0,\ldots,c_{n-1}$ be new constant symbols to $L$ and $L'= L\cup \{c_0,\ldots,c_{n-1}\}$.
   The simple expansion
   $(\A,\alpha_0,\ldots,\alpha_{n-1})$ is a model of $L'$ such that $c_i$ is interpreted by
   $\alpha_i= \alpha(i)$, $0\leq i <n$.
   Instead of $(\A,\alpha_0,\ldots,\alpha_{n-1})$ we shall write shortly $(\A,\alpha)$.

   \begin{thm}\label{labeling}
      Assume $\A$ is a finite model of $L$, $|A|= n$ and $\alpha$, $\beta$ are labelings of $\A$.
      Then the following are equivalent
      \begin{itemize}
         \item[(1)] $(\A,\alpha) \equiv (\A,\beta)$, i.e., $(\A,\alpha)$ and $(\A,\beta)$ are elementary equivalent models,
         \item[(2)] $(\A,\alpha) \cong (\A,\beta)$,
         \item[(3)] $\A_\alpha = \A_\beta$,
         \item[(4)] $\alpha\circ\beta^{-1}\in \Aut(\A)$.
      \end{itemize}
   \end{thm}

   \begin{proof}
     It is well known that finite elementary equivalent models are isomorphic. Hence (1) is equivalent to (2).

     Suppose $(\A,\alpha_0,\ldots,\alpha_{n-1}) \cong (\A,\beta_0,\ldots,\beta_{n-1})$. So there is
     $f\in \Aut(\A)$ such that $f(\beta_i)= \alpha_i$,  $0\leq i <n$. Hence $f\circ\beta=\alpha$, so
     $\alpha\circ\beta^{-1}\in \Aut(\A)$. Therefore (2) implies (4). Reversing this proof,
     it also follows that (4) implies (2).

     Suppose $(\A,\alpha) \equiv (\A,\beta)$ and let $F\in L$ be a $k$-placed function symbol.
     Then  for any choice of constant symbols $c_{i_1},\ldots,c_{i_{k+1}}$,
     $(\A,\alpha)\models F(c_{i_1},\ldots,c_{i_k})= c_{i_{k+1}}$ if and only if
     $(\A,\beta)\models F(c_{i_1},\ldots,c_{i_k})= c_{i_{k+1}}$.
     Hence
     \begin{equation}
        F^\A(\alpha(i_1),\ldots,\alpha(i_k))= \alpha(i_{k+1})\quad \text{\rm iff} \quad
        F^\A(\beta(i_1),\ldots,\beta(i_k))= \beta(i_{k+1}),
     \end{equation}
     therefore
     $\alpha^{-1}(F(\alpha(i_1),\ldots,\alpha(i_k))= \beta^{-1}(F(\beta(i_1),\ldots,\beta(i_k))$
     for all $i_1,\ldots,i_k\in I_n$. Thus we proved that $F^{(\A,\alpha)} = F^{(\A,\beta)}$.
     Similarly we can prove that $R^{(\A,\alpha)} = R^{(\A,\beta)}$ for each
     relation symbol $R\in L$. Hence we proved that $\A_\alpha= \A_\beta$
     and so (1) implies (3).
     Similarly one can prove that (3) implies (1).
   \end{proof}

   Finite models of a first order theory $T$ which have for domains sets $I_n$ are called labeled models of $T$.
   By $\L_{T,n}$ we shall denote the set of all labeled models of $T$ of size $n$.
   By $T_n$ we denote the theory $T\cup \{\sigma_n\}$, where $\sigma_n$ denotes the sentence
   there are exactly $n$ elements. Therefore, $\L_{T,n}$ is the set of all labeled models of $T_n$.

   By a finite theory we mean a first order theory $T$ with finitely many axioms, i.e.,
   $T$ is a finite set of sentences of a finite language $L$.
   We can replace $T$ with a single  sentence, but in some cases we need to add or remove
   a sentence from $T$. In these cases, it is technically easier to work with a set of sentences then
   with a single sentence which replaces $T$.

   Suppose $T$ is a finite theory.
   Let $\P$ be the set of propositional letters defined by \ref{P} over  $A= I_n$ and the language $L$
   and let $T^\ast = \{\varphi^\ast|\, \varphi\in T\}$.
   Further, let $\M(T^\ast) \subseteq 2^\P$ denote the set of  all  models of $T^\ast$,
   i.e., valuations satisfying all propositional formulas in $T^\ast$.
   The following construction describes the correspondence between labeled models of $T$
   and models of $T^\ast$.

   The function $h$ which assigns to each $\mu\in \M(T^\ast)$ a labeled model $h(\mu) = \A$  of $T$ is defined as follows.
   Let $a_1,\ldots,a_k,b\in I_n$. Then

   If $F\in L$ is an $k$-placed function symbol, then
   \begin{equation}\label{Fmu}
      F^{\A}(a_1,\ldots, a_k) = b\quad \textrm{iff}\quad \mu(p_{Fa_1\ldots a_k b}) = 1.
   \end{equation}

   If $R\in L$ is an $k$-placed relation symbol, then
   \begin{equation}\label{Rmu}
     \A\models R[a_1,\ldots, a_k] \quad \textrm{iff}\quad \mu(q_{Ra_1\ldots a_k}) = 1.
   \end{equation}

   By induction on the complexity of the formula $\varphi$, it is easy to prove that $\A\in\L_{T,n}$
   and if ƒ$\mu\not= \nu$, then for the corresponding $\A_\mu$ and $\A_\nu$
   we have $\A_\mu\not =\A_\nu$.   Hence, map $h\colon \M(T^\ast) \to \L_{T,n}$ is one-to-one.
   On the other hand, assume $\A\in\L_{T,n}$. We can use \ref{Fmu} and \ref{Rmu}
   now to define the  valuation $\mu_\A$.
   Since $\A$ is a model of $T$, it follows that $\mu_\A \in \M(T^\ast)$. Hence, $h$ is onto.
   Therefore we proved:
   \begin{thm}
     The map $h$ codes the models in $\L_{T,n}$  by models of $T^\ast$.
   \end{thm}

   This theorem is our starting point in finding finite models of $T$
   of size $n$. As $T^\ast$ is finite, we can replace it with a single
   propositional formula $\theta= \bigwedge_{\psi\in T^\ast} \psi$.
   Obviously, we may consider $\theta$ as a Boolean term $t(\Svm)$.
   Computing $t^{\O_m}(b_1,\ldots,b_m)$ in free Boolean algebra $\O_m$
   for free generators $b_1,\ldots,b_m$,
   we obtain the vector $b$ which by Proposition \ref{prop22}
   codes the full DNF of $\theta$, hence all models of $T^\ast$.
   This gives us all labeled models of $T$ of   size $n$ via the map $h$.

   Let $l_{T,n}$ denote the cardinality of $\L_{T,n}$.
   Obviously, $l_{T,n}$ is equal to the number of bits in  vector $b$ which are equal to 1.

   The mayor target in finite model theory is to count or
   to determine non-isomorphic models of $T$ of   size $n$.
   By $\M(T)_n$ we denote a maximal set of non-isomorphic models of $T$ with
   the domain $I_n$. Elements of this set are also called unlabeled models of $T$.
   By $\kappa_{T,n}= |\M(T)_n|$ we denote the number of non-isomorphic (unlabeled) models of
   $T$ of   size $n$. If a theory $T$ is fixed in our discussion,
   we often omit the subscript $T$ in these symbols. In other words, we shall
   simply write $\L_{n}$, $l_{n}$, $\M_n$ and $\kappa_n$.
   In our examples, the following theorem will be useful in finding numbers $l_{n}$ and $\kappa_n$.
   \begin{thm}{\rm (}{Frobenius - Burnside counting lemma}{\rm )}
      Let $\A$ be a finite model, $|A|= n$.
      Then the number of models isomorphic to $\A$
      which have the same domain $A$ is equal to  $n!/|\Aut(\A)|$.

      If $T$ is a theory of a finite language $L$ with finite number of axioms, then
      \begin{equation}
         l_n = \sum_{\A\in \M_n} \frac{n!}{|\Aut(\A)|}.
      \end{equation}
   \end{thm}

   Note that this theorem immediately follows from theorem \ref{labeling} and   direct
   application of Langrange's subgroup theorem on the symmetric group $S_n$ of $I_n$.

   It is said that a set of models $\K$ is adequate for  $n$-models of $T$ if  $\M_n\subseteq \K\subseteq \L_n$.
   Even for small $n$ the set $\L_n$ can be very large. On the other hand,
   it is possible in some cases  to generate easily all labeled models,
   or to determine $l_n$ from $|\K|$ for
   an adequate family $\K$ of the reasonable size.
   Also, it is commonly hard to generate directly non-isomorphic models of $T$,
   or to compute $\kappa_n$. But for a well chosen adequate set of models these
   tasks can be done. Adequate families are usually generated by filtering
   $\L_n$,  fixing some constants or definable subsets in models of $T$, or
   imposing extra properties, for example adding a new sentences to $T$.
   In our examples  some instances of adequate families will be given.

  \subsection{Killing variables}\label{KillVar}

  Suppose a theory $T$ describes a class of finite models.
  The set of propositional letters $\P$ defined by $\ref{P}$ and which appears in translation from $T$ to
  $T^\ast$  is large even
  for  small $n$ for domains $A=I_n$ from which $\P$ is generated.
  For example, if the language $L$ consists of $k$ unary operations, then
  $|\P|= kn$. If $L$ has only one binary operation $R$, then $|\P| = n^2$.
  If $L$ has only one binary operation $F$, then $|\P|= n^3$.
  Hence, even for small $n$, $\P$ can be enormously large. It can have hundreds,
  or even thousands of propositional variables.
  Hence, we need a way to eliminate some propositional variables appearing
  in  $T^\ast$. Any procedure of elimination variables from $\P$ we shall call killing variables.
  As we have seen, the size of $\P$ which appears in $T^\ast$  and is feasible for computing on
  small computers is bellow $50$ and on supercomputers  below $70$.
  Let us denote by $K$ this feasible number of
  variables\footnote{Hence $50 \leq K\leq 70$ for today's computers}.
  The main goal of killing variables is to reduce $T^\ast$ to a
  propositional theory $T'$ having at most $K$ variables.
  We note that killing variables in general produces
  an adequate set of structures, not  the whole $\L_n$.

  Killing variables  is reduced in most cases by fixing the values of certain variables.
  For example, if $p_{ijk}$ represents a binary operation $i\cdot j=k$, $i,j,k\in A$,
  and if it is known that for some $a,b,c\in A$, $p_{abc}=1$, then for all
  $d\in A$, $d\not= c$, we may take $p_{abd}= 0$.
  The next consideration explains in many cases this kind of killing variables.
  It is related to the definability theory and for notions and terminology we shall
  refer to [\ref{Keisler}].

  Suppose $\A$ is a model of $L$ and $X\subseteq A$. We say that $X$ is absolutely invariant in $\A$
  if for all $f\in \Aut(\A)$, $f(X)\subseteq X$. As usual, $X$ is definable in $\A$ if
  there is a formula $\varphi(x)$ of $L$ so that
  $X= \{a\in A\colon \A\models \varphi[a]\}$.
  The proof of the next theorem is based on the the Svenonius definability theorem,
  cf. [\ref{Svenonius}], or Theorem 5.3.3 in [\ref{Keisler}].

  \begin{thm}
     Let $\A$ be a finite model of $L$ and $X\subseteq A$. Then $X$ is absolutely
     invariant in  $\A$ if and only if $X$ is definable in $\A$.
  \end{thm}

  \begin{proof}
    Obviously, if $X$ is definable then it is absolutely invariant.
    So we proceed to the proof of the other direction.
    In order to save on notation, we shall take $L=\{R\}$,
    $R$ is a binary relation symbol.
    Suppose $X$ is invariant under all automorphisms of $\A$. Let
    $\psi_1(U)$ be the following sentence of $L\cup \{U\}$, $U$ is a new
    unary predicate:
    \begin{equation}
       \begin{array}{l}\displaystyle
          \forall x_1\ldots x_n\forall y_1\ldots y_n
          (
            ( \bigwedge_{i<j} x_i\not= x_j \wedge \bigwedge_{i<j} y_i\not= y_j \wedge
              \bigwedge_{i,j}(R(x_i,x_j)\Leftrightarrow R(y_i,y_j)) )              \\
            \phantom{A} \hskip 1cm
               \displaystyle\Rightarrow \bigwedge_i (U(x_i)\Rightarrow U(y_i))
          ).
       \end{array}
    \end{equation}

     The sentence $\psi_1(U)$ states that $U$ is absolutely invariant in any model $\B$ of $L$
     which has $n$ elements,
     i.e., if $(\B,Y)\models \psi_1(U)$ then $Y$ is absolutely invariant in $\B$.

     Let $\psi_2$ be the following sentence of $L$:
     \begin{equation}\label{psi2}
          \displaystyle
          \exists x_1\ldots x_n
          (
            \bigwedge_{i<j} x_i\not= x_j \wedge \forall x \bigvee_i x= x_i \wedge
              \hskip -2mm \bigwedge_{R^\A(i,j)}\hskip -3 mm R(x_i,x_j) \wedge
              \hskip -3mm \bigwedge_{\neg R^\A(i,j)} \hskip -4.5mm \neg R(x_i,x_j)
          ).
     \end{equation}

  We see that the sentence $\psi_2$ codes the model $\A$, i.e., if $\B$ is a model of $L$ and
  $\B\models \psi_2$ then $\B\cong \A$.

  Let $\psi(U)= \psi_1(U)\wedge \psi_2$. Suppose $\B$ is any model of $L$, $(\B,Y)$ and $(\B,Y')$ are
  expansion of $\B$ to models of $\psi(U)$ and assume $(\B,Y)\cong (\B,Y')$.
  Then we see that $Y=Y'$. Therefore,  by Svenonius theorem it follows
  that $\psi$ defines $U$ explicitly up to disjunction.
  In other words there are formulas $\varphi_1(x),\ldots,\varphi_m(x)$
  of $L$ such that
  \begin{equation}
     \psi(U)\models \bigvee_i \forall x(U(x) \Leftrightarrow \varphi_i(x))
  \end{equation}

  As $(\A,X)\models \psi(U)$, there is $i$ so that
  $ (\A,X)\models \forall x(U(x) \Leftrightarrow \varphi_i(x))$.
  Hence $X$ is definable by $\varphi_i(x)$.
  \end{proof}

  The following corollaries follow by direct application of the last theorem
  to one-element absolutely invariant subsets.
  \begin{cor}
    Let $\A$ be a finite model of finite $L$ and $a\in A$. If $a$ is fixed by all automorphisms of $\A$
    then $a$ is definable in $\A$ by a formula $\varphi(x)$ of $L$.
  \end{cor}
  \begin{cor}
    Let $\A$ be a finite model of finite $L$.  Then
      $\Aut(\A)= \{i_A\}$ if and only if
      every element of $A$ is definable in $\A$.
  \end{cor}

  Here are  other examples of absolutely invariant, and hence definable subsets $X$ in various types of
  finite structures $\A$. If $\sim$ is a relation of equivalence over $A$ and $k\in N$, then
  $X=$ "the union of all classes of equivalences of   size $k$" is absolutely invariant.
  Let $\A= (A,\leq)$ be a partial order. Then the set $S$ of all minimal elements and
  the set $T$ of all maximal elements of $\A$ are absolutely invariant.
  The same holds for the set of all minimal elements of $A\backslash S$.
  In groups, characteristic subgroups, such as the center and the commutator subgroup,
  are absolutely invariant.


  In our examples we shall often use the following argument.
  Let $T$ be a finite theory of $L$ and
  assume $\varphi_0(x),\ldots,\varphi_{k-1}(x)$ are  formulas of $L$ for which  $T$
  proves they are  mutually disjoint, i.e., for $i\not = j$,
  $T\vdash \neg \exists x(\varphi_i(x)\wedge \varphi_j(x))$. Assume they define
  constants in $T$, in other words,  for each $i$
  \begin{equation}
    T\vdash \exists x(\varphi(x)\wedge \forall y (\varphi_i(y) \Rightarrow \varphi_i(x)).
  \end{equation}

  Let $\B$ be a model of $T$, $B= I_n$ and $\B\models T$.
  Then $\B$ has a unique expansion to $(\B,b_0,\ldots,b_{k-1})$ which is a model of
  $T'= T\cup \{\varphi_0(c_0),\ldots,\varphi_k(c_{k-1})\}$,  $c_0,\ldots,c_{k-1}$ are
  new symbols of constants to $L$. Since $b_i\not= b_j$ for $i\not= j$ we can define
  \begin{equation}\label{fIk}
     f\colon I_k \to \{b_0,\ldots,b_{k-1}\},\quad f(i)= b_i,\quad i=1,\ldots,k.
  \end{equation}

  It is easy to see that we can define labeled model $\A$ of $L$
  and that $f$ extends to $h\colon (\A,0,\ldots,k-1) \cong (\B,b_0,\ldots,b_{k-1})$.
  Hence, for an adequate set of $n$-models of $T$
  we can choose a set $\K$   of labeled models $\A$ of $T$
  such that $(\A,0,\ldots,k-1)$ is a model of $T'$. Therefore, models in $K$ have
  the fixed  labelings by $0,\ldots,k-1$ of constants definable in $T$.

  Obviously, we can take in \ref{fIk} any $S\subseteq I_n$, $|S|=k$ instead
  of $I_k$. There are $\binom{n}{k}$ such choices of $S$.
  Let $s$ denote a permutation $s_0\ldots s_{k-1}$ of $S$ and
  $\K_s$ the corresponding adequate set for $n$-models of $T$:
  for models $\A$ in $\K_s$, $(\A,s_0,\ldots,s_{k-1})$ is a model of $T'$.
  In other words, definable elements formerly labeled by $0,\ldots,k-1$
  in models of $\K$ they are labeled now in $\K_s$ by $s_0,\ldots,s_{k-1}$.
  Suppose  $S$ and $S'$ are $k$-subsets of $I_n$  and
  $s$, $s'$  permutations either of $S$ or $S'$, $s\not= s'$.
  Then $\K_s\cap \K_{s'}= \emptyset$ and $|\K_s| = |\K_{s'}|$.  Hence
    $\L_{T,n}= \bigcup_s \K_s$
  and so
  \begin{equation}
     l_{T,n}= \binom{n}{k}k!|\K| = n(n-1)\cdots (n-k+1) |\K|.
  \end{equation}

  In many cases   theory $T$ determines the values of atomic formulas
  which contains some of the definable constants. Hence, the corresponding
  propositional letter from $\P$ has a definite value.
  For example, suppose $R$ is a 2-placed relation symbol and that $T$ proves
  $\forall x R(c_0,x)$. Then we can take $p_{0i}=1$, $i= 0,\ldots,n-1$. Hence,
  if $\P$ is generated over $I_n$,  $n$ propositional variables
  are killed  in $\P$. The remaining number of variables is $n^2-n$.


  \subsection{Definable partitions}\label{dp}

   The presented idea with definable constants can be extended to definable subsets as well.
   For simplicity, we shall assume that $L= \{R\}$, where $R$ is a binary relation symbol.

   A sequence $\Delta= \theta_1(x),\ldots,\theta_m(x)$ of formulas of $L$ is called a definable partition for $T_n$ if
   $T_n$ proves:

   \begin{itemize}
      \item[1.] $\forall x(\theta_1(x)\vee \ldots \vee \theta_m(x))$.
      \item[2.] $\neg \exists x (\theta_i(x)\wedge \theta_j(x))$,\quad $1\leq i \leq j \leq m$.
   \end{itemize}

   We shall say that $\Delta$ is a good definable partition if there are formulas
   $S_{ij}(x,y)$,  $1\leq i,j\leq m$, such that each $S_{ij}(x,y)$ is one of
   $R(x,y)$, $R(y,x)$, $\neg R(x,y)$, $\neg R(y,x)$,
   and $T_n$ proves:
   \begin{equation}\label{good}
     \forall xy ((\theta_i(x) \wedge \theta_j(y)) \Rightarrow S_{ij}(x,y)),\quad 1\leq i \leq j \leq m.
   \end{equation}

   \begin{exm}
   It is easy to write first-order formula $\theta_k(x)$  which says  that $x$ has
   exactly $k$ R-connections with other elements. In other words, $\theta_k(x)$
   expresses that there  are exactly $k$ elements $y$ such that $R(x,y)$.
   Assume $T_n$ proves that $R$ is an acyclic graph. Then $k\leq l$ implies
   $(\theta_k(x)\wedge\theta_l(y))\Rightarrow \neg R(x,y)$. Hence, in this case definable
   partition $\theta_k(x)$ is good.
   \end{exm}

   In any labeled model $\A$ of $T_n$, $\Delta$ determines  sequence $\X$ of definable subsets
   $X_1,\ldots, X_m$. By a component we shall mean elements of $\X$.
   It may happen that some components are empty.
   The sequence of non-empty sets from
   $\X= (X_1,\ldots,X_m)$ form an ordered partition of $A$.
   A sequence $\X$ with this property will be called a $c$-partition.

   Our idea for using a good definable partition $\Delta$ in generating labeled models $\A$ of $T_n$
   is as follows. We assume that the propositional letter $p_{ij}$ represents
   $R^\A(i,j)$ as described by \ref{Rmu}.
   We generate all $c$-partitions $\X= (X_1,\ldots,X_m)$ of $I_n$ that are potentially components
   of $\A$,   taking that $X_i$ corresponds to $\theta_i$.
   For each $\X$ we assign values to particular $p_{ij}$ in the following way.
   If $S(x,y)$ is $R(x,y)$ then we set $p_{ij}=1$ for $i\in X_k$ and $j\in X_l$, $k\leq l$
   and if $S(x,y)$ is $\neg R(x,y)$, then we set $p_{ij}=0$. We assign similarly
   values to $p_{ij}$ if $S(x,y)$ is $R(y,x)$ or $\neg R(y,x)$.
   Therefore we obtained propositional theory $T_\X\subseteq {T^\ast}_{\hskip -2mm n}$
   with the reduced number of  unknowns from $\P$.
   Then set $\K_\X$ of labeled models corresponding to $T_\X$ in the sense of
   Subsection \ref{correspondence} is adequate for  set $\L_\X$ of all labeled models
   of $T_n$ in which $\Delta$ defines partition $\X$.

   Obviously, every model of $T_n$ is  isomorphic to a model $\A$ with domain $I_n$
   with canonical components $\X$
   \begin{equation}\label{c2}
      \begin{array}{l}\displaystyle
        X_1= \{0,1,\ldots,\alpha_1-1\},\, X_2= \{\alpha_1,\alpha_1 +1,\ldots,\alpha_1 + \alpha_2-1\},\ldots, \\
        \phantom{.}\hskip 5mm   X_m= \{\displaystyle\sum_{i<m} \alpha_i,\sum_{i<m} \alpha_i + 1,\ldots,\sum_{i\leq m} \alpha_i - 1\}.
      \end{array}
   \end{equation}

   Let us denote by $\PP$ the set of all $c$-partitions of $I_n$.
   Then every model $\A\in\L_\X$, $\X=(X_1,\ldots,X_n)$ is obtained from a model $\B\in \K_X$
   choosing component $X_1$ from $I_n$, then $X_2$ from $I_n\backslash X_1$, $X_3$ from $I_n\backslash \{X_1\cup X_2\}$
   and so on, until all $X_i$ from $\X$ are exhausted.
   Therefore
   \begin{equation}\label{ltn}
      l_{T,n}= \sum_{\X\in \PP}  \binom{\beta_1}{\alpha_1}\ldots \binom{\beta_k}{\alpha_k} |\K_\X|
   \end{equation}
   where $\X= (X_1,\ldots,X_k)$, $|X_i|$= $\alpha_i$ and
   \begin{equation}
      \beta_1= n,\quad \beta_2= \beta_1- \alpha_1,\quad \ldots,\quad \beta_k= \beta_{k-1} - \alpha_{k-1}.
   \end{equation}

   Note that if $\X \not= \Y$, $\X, \Y \in \PP$, and if $\A\in \K_X$ and $\B\in \K_Y$, then
   $\A$ and $\B$ are non-isomorphic. Hence, if $\kappa_{\X,n}$ is the number of
   non-isomorphic models in $\K_\X$, then
   \begin{equation}\label{kappan}
     \kappa_n = \sum_{\X\in \PP} k_{\X,n}.
   \end{equation}

   The following proposition is useful in estimation of the number of computing steps of $\K_\X$.

   \begin{prop}
      Assume $|A|=n$. Then there are
      \begin{equation}\label{cp}
        c_{nm}=\sum_{k=1}^m \binom{m}{k}\binom{n-1}{k-1}
      \end{equation}
      $c$-partitions $\X= (X_1,\ldots,X_m)$ of $A$.
   \end{prop}
   \begin{proof}
     Let $|X_i|= \alpha_i$.
     Therefore  $\alpha_1,\ldots,\alpha_m$ is an integer solution of
     \begin{equation}\label{eq2}
       n= x_1 + \ldots + x_m,\quad x_1,\ldots,x_m\geq 0.
     \end{equation}
     Since the integer solutions of
     \begin{equation}\label{eq3}
       n= x_1 + \ldots + x_k,\quad x_1,\ldots,x_k\geq 1.
     \end{equation}
   are obtained from \ref{eq2} by choosing $k$ variables $x_i\not=0$, $k\geq 1$, and \ref{eq3} has $\binom{n-1}{k-1}$ solutions,
   there are $\binom{m}{k}\binom{n-1}{k-1}$ solutions of \ref{eq2}. Hence, there are
   $\binom{m}{k}\binom{n-1}{k-1}$ $c$-partitions $\X$ with exactly $k$ nonempty sets $X_i$.
   Summing up for $k=1,\ldots,m$, we obtain expression \ref{cp} for $c_{nm}$.
   \end{proof}

   Hence, for an adequate set of models of $T_n$ we can take set $\K$ of labeled models $\A$ of $T_n$ with the
   components \ref{c2}. So our method  for computing models of $\K$ is as follows.
   As usual, the propositional letter $p_{ij}$ stands for $R(i,j)$.
   \medskip

   {\bf Counting procedure TBA}

   \begin{itemize}
     \item[1.] Find   good definable partition $\theta_1(x),\ldots,\theta_m(x)$
               which satisfies condition \ref{good}.
     \item[2.] Generate all $c$-partitions $\X= (X_1,\ldots,X_m)$ of $I_n$ with arrangements \ref{c2}.
     \item[3.] {\it Killing variables}: For all $1\leq k \leq l \leq n$ we
               fix the values of certain $p_{ij}$ as follows.
               Take $p_{ij}=1$ for $i\in X_k$ and $j\in X_l$ if $S(x,y)$ is $R(x,y)$.
               If $S(i,j)$ is $\neg R(i,j)$ then we take $p_{ij}= 0$.
               If $S(i,j)$ is $R(j,i)$, then $p_{ji}= 1$.
               If $S(i,j)$ is $R\neg R(j,i)$, then set $p_{ji}=0$.
     \item[4.] Reduce $T_n^\ast$ to $T_\X$ with the reduced number of variables using assigned values
               to variables $p_{ij}$ in the previous step.
     \item[5.] Generate and count models of $\K_\X$ using $T_\X$ and free Boolean vectors by the  procedure
               described in  section \ref{correspondence}.
     \item[6.] Find $\kappa_{X,n}$  by enumerating elements of $\K_X$.
     \item[7.] Repeat steps (5) and (6) until $\PP$ is exhausted.
     \item[8.] Compute $l_{T,n}$ by formula \ref{ltn}.
     \item[9.] Compute $\kappa_{T,n}$ by \ref{kappan}.
   \end{itemize}


  \section{Program implementation}

  We implemented the algorithms and ideas presented in the previous sections
  into a programming system which we shall call TBA. It is divided into
  two layers. The first one is implemented in OpenCL which we have chosen as
  a good framework for writing parallel programs that execute across heterogeneous
  platforms consisting of central processing units (CPUs) and graphics processing units (GPUs).
  This part of code manipulates with free Boolean vectors as described in subsection \ref{CBe}
  and it is invisible to the general user of TBA.
  The second layer is developed in Python programming language  and we used it to achieve  two goals.
  The first-one is to manipulate Boolean expressions as described in subsections
  \ref{Translation} and \ref{KillVar}. The second aim was to define new constructs in Python
  mainly related to the predicate calculus. The general user may use them into scripts
  to solve combinatorial problems using techniques such as described in
  subsection \ref{dp}. The main body of a script strictly follow the syntax of
  predicate calculus, but Python standard constructs can be embedded in the scripts as well.
  The user  executes the scripts  by TBA in the terminal mode.

  \subsection{TBA Core}

  The core of the system is a parallel computational engine that searches for models
  of a  Boolean formula $\varphi(x_1,\ldots,x_n)$.
  This part of TBA is generated in OpenCL language which is then compiled to binaries and executed.
  In a sense, the core uses brute force search over a problem space,
  but utilizing all of the available  bit level parallelism of the underlying hardware
  as described in Subsection \ref{CBe}.
  Whenever $\varphi(x_1,\ldots,x_n)$ is dispatched to the engine, it first partitions search space  $S$.
  The table of $S$ associated to $\varphi(x_1,\ldots,x_n)$
  is  of size $n\times 2^n$ and consists of free Boolean vectors.
  The partitioning of $S$ is done by slicing  this table into appropriate blocks and
  depends on the number of available processors and  memory.
  Due to the simplicity of the representation, the slicing scheme is
  very scalable. This enables us to choose a partition such that all
  cores of all of available processing units are used in parallel in further computation.

  In addition, the engine generates an efficient computing tree for $\varphi(x_1,\ldots,x_n)$,
  adapted to the actual parallel hardware  and hardware architecture.
  The implementation is done for both, GPU's and CPU's and it is on the user
  which implementation will be used.
  While for GPU's the advantage is the number of computing cores,
  for CPU's this is the length of the vector units and the processor's
  speed. Modern GPU's have more than 2000 computing 32 bit cores, while,
  in contrast, CPU's have four cores, 256 bit registers and
  up to four time faster clock speed.
  The approximative formula for the ratio between the speeds of the execution
  of our code on a GPU and on a CPU is:
  \begin{equation}\label{gpucpu}
     f= \frac{n_g b_g s_g}{n_c b_c s_c}
  \end{equation}
  where $n_g$ is the number of $b_g$-bit computing cores and
  $s_g$ is the number of clock cycles of GPU, while $n_c, b_c, s_c$ are the similar
  parameters for the CPU ($b_c$ is the number of bits of the vector unit).
  Hence, for the above mentioned configuration ($n_g= 2^{11}$, $b_g= 2^5$,
  $n_c = 2^2$, $b_c= 2^8$ and $s_c/s_g = 4$), we have $f=16$. Therefore,
  GPU's are superior to CPU's and our tests are in agreement with \ref{gpucpu}.

  There are also other submodules. Submodule {\tt Translate} translates
  predicate formulas into Boolean expressions according to the rules explained
  in Subsection  \ref{Translation}. It also build the computing tree
  of so obtained Boolean term. Another important submodule is
  {\tt Reduction} which reduces a Boolean term having constants $0$ and $1$
  to the expression without these constants. We observe that a Boolean
  expression may have several hundreds of thousands of characters, but
  {\tt Reduction} is limited not by the size of the expression,
  but only by the available computer's memory.

  \subsection{TBA scripts} TBA scripts are used to implement algorithms for
  generating and counting finite combinatorial structures such as specific graphs,
  orders, Latin squar\-es, automorphisms of first-order structures, etc.
  The user writes TBA scripts as txt files and they follow Python syntax.
  In general their structure consist of three parts.
  The first part contains definitions of domains over which combinatorial objects are generated.
  The second one consists of definitions of combinatorial structures by axioms written in the
  syntax of the predicate calculus. Propositional calculus is embedded into Python, but we had to
  expand it with bounded quantifiers in order to express predicate formulas having
  in mind finite structures as the main (and only) semantics.
  The quantifier extension of Python we named Python-AE, since
  we denoted by $A$ the universal quantifier  and the existential quantifier by $E$.
  Finally, the third part is used for killing variables, as described in
  subsection \ref{KillVar}. These parts are not strictly separated and they may overlap.

  Here is a simple example of a TBA script, named {\tt SO.txt}.
  It computes all partial orders over domain $S= \{0,1,\ldots,n-1\}$
  with a special element. The propositional letter $p_{ij}$ stands for $i\leq j$.
  An element $a\in S$ is special if it is comparable with all elements of domain $S$.
  \smallskip

  {\tt
    n= 6

    S= range(n)

    S2= perm(range(n),2)

    S3= perm(range(n),3)
    \smallskip

    f1= A[i,j:S2] ($\sim$p(i,j) | $\sim$p(j,i))

    f2= A[i,j,k:S3] ($\sim$(p(i,j)\& p(j,k)) | p(i,k))

    f3= E[i:S].A[j:S] (p(i,j) | p(j,i))
    \smallskip

    assumptions= $\{$p(i,i): 1 for i in S$\}$
  }
  \smallskip

  First four lines define domain $S= \{0,1,2,3,4,5\}$, set $S_2$ of ordered pairs of elements of $S$
  with distinct coordinates and $S_3$, the set of triplets. The next three lines define predicate formulas
  $\varphi_1, \varphi_2, \varphi_3$.
  Boolean operation signs are represented in the standard Python notation. Hence, the Python signs
  $ \sim, \&, |, \,\, \hat{}\,\, $
  stand respectively for $\neg, \wedge, \vee, +$, where
  $x+y = x\bar y \vee \bar x y$ (symmetric difference of $x$ and $y$).
  The construct $A[i:S]$ stand for the bounded universal quantifier  (in the manner of
  Polish logic school, eg [\ref{Mostowski}]) $\bigwedge_{i\in S}$. Similarly, $E[i:S]$ denotes
  the bounded existential quantifier $\bigvee_{i\in S}$.
  Hence, $f_1, f_2, f_3$ are Python-AE transcripts of the following predicate formulas,
  if $p_{ij}$ is read as $i\leq j$:
  \begin{equation}\label{SO}
    \begin{array}{l}\displaystyle
      \varphi_1= \bigwedge_{i,j\in S_2} (\neg p_{ij} \vee \neg p_{ji})    \\
      \displaystyle
      \varphi_2= \bigwedge_{i,j,k\in S_3}  (\neg(p(i,j) \wedge p(j,k)) \vee p(i,k)) \\
      \displaystyle
      \varphi_3= \bigvee_{i\in S} \bigwedge_{j\in S} (p_{ij} \vee p_{ji}).
    \end{array}
  \end{equation}

  Obviously, $\varphi_1$ states that  $\leq$ is antisymmetric ie,
  $\forall i,j\in S (i\leq j\wedge j\leq i \Rightarrow i=j)$.
  Further, $\varphi_2$ states that $\leq$ is transitive,
  $\forall i,j,k \in S(i\leq j\wedge j\leq k \Rightarrow i\leq k)$,
  assuming it is reflexive.
  The reflexivity is handled in the last line of the script.
  Finally, $\varphi_3$ states that the order has a special element, $\exists i\in S\, \forall j\in S (i\leq j \vee j\leq i)$.

  The last line states that $\leq$ is reflexive.
  It also kills variables $p_{ii}$, $i\in S$.
  The last line can be replaced by $\bigwedge_{i \in S}$,
  but during the execution of {\tt SO.txt} we would have then more
  free variables and the program would be less capable.
  Observe that there are all together $n^2$ variables $p_{ij}$ and that $n$
  variables are killed. Hence, during the execution of the script,
  there are $n^2 -n$ free variables.
  Our current implementation solves on GPU's systems of the Boolean equations which have
  up to 30 unknowns and on CPU's with up to 32 unknowns.
  Hence the script can be run for $n\leq 6$.
  More sophisticated examples which could be executed for much larger $n$
  are explained in the next section.

  The script is executed on a GPU (default case) by\,
  {\tt solve.exe --all SO.txt}

  \noindent
  and on a CPU:\,
  {\tt solve.exe --all --cpu SO.txt}.

  Output file {\tt out.txt} contains after execution all solutions of \ref{SO},
  ie, all models of propositional formulas which are $\ast$-transforms of
  formulas $\varphi_1, \varphi_2, \varphi_3$ in the sense of subsection \ref{Translation}.
  All partial orders $(S,\leq_\mu)$ with a special element are obtained then
  by choosing valuations (rows) $\mu$ from {\tt out.txt} and setting
  $i\leq_\mu j$ iff $\mu(p_{ij})=1$.

  Here are some general remarks and basic rules for Python-AE.
  Predicate formulas only with bounded quantifiers are allowed and must be in
  written in the prenex normal form. The quantifier-free part otherwise follows
  the Python syntax for Boolean expressions and must be parenthesized.
  Quantifiers are delimited from each others by the dot sign.

  Killing variables means setting values for some variables
  appearing in formulas of a TBA script file.
  Construction implemented in Python for killing variables is called {\tt assumptions}.
  Assumptions for killing variables are defined using Python dictionary structure.
  For example {\tt $\{$a:1, b:0$\}$} defines a dictionary which sets values of two variables: $a= 1, b= 0$.
  In this way, listing values of variable,
  any dictionary for killing variables can be constructed.
  A dictionary can be constructed also in other ways using Python syntax.
  \smallskip

  \noindent
  {\bf Example} (dictionary comprehension): {\tt assumptions= $\{$p(i):\,1 for i in S$\}$}.
  In this way we defined  {\tt p(i)=1} for all $i$ in $S$.
  \smallskip

  Dictionary which defines values of variables must be named {\tt assumptions}.
  The above example demonstrates killing variables using incremental method
  applied on assumptions (dictionary):
  An already existing dictionary (assumptions) is updated by
  the command {\tt assumptions.update}.
  If {\tt assumptions.update} refers to already killed variables, their values
  are set to new values defined by this command. Hence, the order of updating is important.

  A TBA script file {\tt file.txt} is executed in
  the terminal mode by {\tt solve.exe --all file.txt}. The result of
  the execution is placed in {\tt out.txt}.

  \section{Examples}

  The portable codes for executing programs in our system, explanation how to use them and all examples
  described in this paper and some additional ones, can be found at the address
  http://www.mi.sanu.ac.rs/$\sim$pejovica/tba.
  Most of our examples are tested against to the examples
  from the On-Line Encyclopedia of Integer Sequences (OEIS)\footnote{http://oeis.org}.
  In all cases, our results were in the agreement with the results which we found there.

  \subsection{Solving Boolean equations}

  Solving Boolean equations is the simplest use of our software.
  Any system of Boolean equation should be written in our system in the following way:
  \begin{equation}
    e_1= \varphi_1(x_1,\ldots,x_n),\ldots, e_k= \varphi_k(x_1,\ldots,x_n).
  \end{equation}
  The program finds all $(\alpha_1,\ldots,\alpha_n)\in 2^n$ such that
  $\varphi_(\alpha_1,\ldots,\alpha_n)\equiv 1$, $1\leq i \leq k$.

  Here is an
  example of two Boolean equations with unknowns $x, y, z, u$ (example BAequ4\_in.txt
  at the above address):
  \begin{equation}\label{BAequ}
     x + y + \bar z + u = 1, \quad  x \vee yz = u
  \end{equation}
  The second equation is equivalent to $\neg((x \vee yz) + u)= 1$.
  Hence, Python-AE file BAequ4\_in.txt
  solves \ref{BAequ} and contains only two lines:
  \begin{equation}
    \begin{array}{lll}
        e1 &= &x\,\, \hat{}\,\, y \hskip 1mm\hat{} \sim\hskip -1mm z \,\, \hat{}\,\, u  \\
        e2 &= &\sim((x\hskip 1mm |\hskip 1mm y\hskip 1mm\&\hskip 1mm z)\,\, \hat{}\,\, u)
    \end{array}
  \end{equation}
  File BAequ4\_in.txt is executed on a GPU (default case) by
  \smallskip

  {\tt solve.exe --all BAequ4\_in.txt BAequ4\_out.txt}.
  \smallskip

  \noindent
  and on a CPU by
  \smallskip

  {\tt solve.exe --all --cpu BAequ4\_in.txt BAequ4\_out.txt}.
  \smallskip

  \noindent
  Output file BAequ4\_out.txt contains after execution all solutions of \ref{BAequ}.

  \subsection{Ordered structures}\label{os}

    Let $T$ be the theory of partial orders of $L= \{\leq\}$ having at least
    $2$ elements with extra axioms which state there are
    the least  element and the greatest element  x.
    Instead of $T$ we can take the theory $T_1$ of partially ordered sets which are
    upward and downward directed. Theories $T$ and $T_1$ are not equivalent, for example
    $T_1$ has an infinite model which is not a model of $T$. But $T$ and $T_1$ have
    same finite models.

    We see that  $l_{T,n}= n(n-1) |\K|$, $n\geq 2$, where $\K$ is the set
    of all partial orders $\A=(A,\leq,0,n-1)$, $A=I_n$, $0$ is the least and $n-1$ is the greatest element in $\A$.
    Since $p_{ij}$ states $i\leq j$ and $\leq$ is reflexive, we can also take ($n\geq 2$)
    \begin{equation}\label{conspa}
       \begin{array}{l}
          p_{0i}= 1,\,\, p_{j0}= 0,\,\, p_{i1}=1,\,\, p_{1k}= 0,\,\, p_{ii}=1,   \\
          i=0,\ldots,n-1,\,\, j= 1,\ldots,n-1,\,\, k= 0,\ldots, n-2.
       \end{array}
    \end{equation}
    Hence, $5n-6$ variables are killed and $T^\ast$ is reduced to $T'$ which has
    $v= n^2 - 5n +6$ variables. If $n=8$ then $v= 30$ and all partial orders having $8$
    elements are generated in  one computer cycle in our computer installation.
    Simply adding to $T$ some new axioms, we can generate  models of the new theory in
    the same way and the same computing time.
    For example, in this way we can compute all lattices  of order $8$ just by adding to $T$ only one axiom.

    With small adjustments, this algorithm works on small computers in real time for $n\leq 12$.
    Namely, for larger $n$, the feasibility constant $K$,
    see the footnote (3), is exceeded. For larger $n$ we have to use the previously described procedure
    based on components.
    In order to describe them, let us define recursively the following sequence of length $n$ of the following formulas.
    \begin{equation}
       \theta_0(x) \equiv \forall y(x\leq y),\quad
       \theta_{k+1}(x)\equiv \forall y(\bigvee_{i\leq k}\theta_i(y) \vee x\leq y)\wedge \bigwedge_{i\leq k}\neg \theta_i(x).
    \end{equation}

    If $\A= (A,\leq)$ is a partial order with domain $I_n$, we see that the associated components are:
    $X_0= \{0\}$, $0$ is the least element of $\A$, $X_1$ is the set of minimal
    elements of $A\backslash \{0\}$, $X_2$ is the set of minimal elements of
    $A\backslash (X_0\cup X_1)$, and so on. Let us call an element of layer $X_k$,
    a $k$-minimal element. Since $X_{i+1}\not=\emptyset$ implies $X_i\not=\emptyset$,
    we see that $X_k=\emptyset$ for $k>m$ for some $m\leq n$. Hence,
    \begin{equation}\label{cpA}
       \X= (X_1,\ldots,X_m, 0,\ldots,0),\quad X_i\not= \emptyset,
    \end{equation}
    is the associated $c$-partition of $A$.

    \begin{prop}\label{posn}
         Let $\A$ be a partial order of size $n$ with the least element and the greatest element.
         Then the number of $c$-partitions {\rm (\ref{cpA})} of $A$ which consist from layers $X_k$ of
         $k$-minimal elements is   $c_n= 2^{n-3}$.
    \end{prop}
    \begin{proof}
      Obviously, the least element and the greatest element can be omitted from $A$.
      Hence, we count $c$-partitions of $A'= \{1,2\ldots,n-2\}$.
      Let $|X_i|= \alpha_i$.
      Therefore  $(\alpha_1,\ldots,\alpha_m)$ is an integer solution of
      \begin{equation}\label{eq4}
        n-2 = x_1 + \ldots + x_m,\quad x_1,\ldots,x_m\geq 1,
      \end{equation}
      where  $m\leq n-2$. Equation \ref{eq4} has $\binom{n-3}{m-1}$ solutions, hence the total
      number of $c$-partitions \ref{cpA} over domain $A$ is
      \begin{equation}
        c_n = \sum_{m=1}^{n-2} \binom{n-3}{m-1} = 2^{n-3}.
      \end{equation}
    \end{proof}
    If $i\in X_k,\, j\in X_l,\, l \leq k$ then $i\not\leq j$.
    Hence, in addition to (\ref{conspa}), for each $c$-partition  more variables $p_{ij}$ are killed:
    \begin{equation}
        p_{ij}= 0,\quad i\in X_k,\, j\in X_l,\, l \leq k.
    \end{equation}

    For so introduced parameters, we can use the counting procedure TBA
    (Section \ref{dp}) for finding and counting labeled and unlabeled
    partial orders of size $n$ with the least element and  the greatest element.
    According to Proposition \ref{posn}, the procedure consists from $2^{n-3}$
    loops. In each loop, a $c$-partition $\X= (X_1,\ldots,X_m, 0,\ldots,0)$, $X_i\not= \emptyset$,
    is produced and adequate family $\K_\X$ from which labeled and unlabeled models
    are generated and counted by (\ref{ltn}) and (\ref{kappan}).
    A program implementation of this procedure in our system can be found at the given above address.
    Simply adding axioms for particular types of ordering, e.g. lattices,
    distributive  lattices, etc. we construct and count labeled and
    unlabeled structures of this particular type as well.

  \subsection{Other examples}

   Semantics of our system lay in the first order predicate logic, hence
   in principle models of any class of finite structures described in
   this logic can be computed. The obvious limitation is
   the memory size and the hyper-exponential growth of the number
   of propositional variables appearing in the description of the
   related class of models. However, with a good choice of an adequate
   subclass of models and the ably reduction (killing)
   of variables we believe that new and interesting results in computational discrete mathematics
   can be obtained. Even if the aim of this paper is not to
   study the particular class of finite structures,
   we proposed a number of examples of this kind.
   These examples refer to ordered structures,
   automorphisms of structures and Latin squares (quasigroups).
   Examples of interest include computations of various types of
   lattices and a solution of Sudoku problem.
   In Sudoku problem appear 729 propositional variables,
   but our system solved it effortlessly by virtue of good elimination (killing) of
   variables.
   There are particular attempts for analysis and modeling classes of Latin squares
   in propositional calculus, eg [\ref{Ercsey}], [\ref{Lewis}] and
   [\ref{Zhang}]. In contrast to our approach, their computation relies on Davis-Putnam algorithm.
   Our aim is to refine some of the ideas we have just outlined,
   particularly based on definability as presented in subsections
   \ref{KillVar}, \ref{dp} and Example \ref{os}.

{}


\begin{thebibliography}{}

 \bibitem{1}\label{Burris} S. Burris and H.P. Sankappanavar,
 {\it A course in Universal algebra}, Springer, 1981, 2012 Update.

 \bibitem{2}\label{Dow} A. Dow, P. Nyikos, {\it Representing free Boolean algebras}, Fundamenta Mathematicae,
 \textbf{141}, (1992), 21--30.

 \bibitem{2a}\label{Ercsey} Maria Ercsey-Ravasz, Zoltan Toroczkai, {\it The Chaos Within Sudoku},
 Scientific Reports 2, Article number: 725 doi:10.1038/srep00725, 2012.

 \noindent
 http://www.nature.com/srep/2012/121011/srep00725/full/srep00725.html

 \bibitem{3}\label{Keisler} C.\,C. Chang, J.\,H. Keisler, {\it Model theory}, North Holland, (1990).

 \bibitem{4}\label{Lewis} Rhyd Lewis,  {\it Metaheuristics can Solve Sudoku Puzzles},
 Journal of Heuristics,  Vol. 13, Issue 4, pp 387-401, 2007.

 \bibitem{5}\label{Mostowski} K. Kuratowski, A. Mostowski, {\it Set Theory}, PWN, 1967.

 \bibitem{6}\label{MijajloFB} \v Z. Mijajlovi\'c, {\it On free Boolean vectors}, Publ. Inst. Math, \textbf{64(78)}, 1998, 2--8.

 \bibitem{7}\label{MijajloBorel} \v Z. Mijajlovi\'c, D. Doder, A. Ili\'c-Stepi\'c, {\it Borel sets and countable models},
    Publ. Inst. Math, \textbf{90(104)}, (2011), 1--11.

 \bibitem{8}\label{MijajloMT} \v Z. Mijajlovi\'c, {\it Model Theory}, Novi Sad, (1985).
 \bibitem{9}\label{Svenonius} L. Svenonius {\it A theorem on permutations in models}, vol. 25, 173-178, 1959.

 \bibitem{10}\label{Zhang}  Hantao. Zhang, Maria Paola Bonacina, Jieh Hsiang
 {PSATO: \it a Distributed Propositional Prover and its Application to Quasigroup Problems},
 Jour. of Symbolic Computation, Vol. 21, Issues 4–6, 1996, 543–560.

 \bibitem{11}\label{Sikorski} R. Sikorski, {\it Boolean Algebras}, Springer-Verlag, Berlin, (1969).



 \end{thebibliography}
\end{document}